\documentclass[12pt,a4paper,twoside]{article}

\usepackage{amsmath}
\usepackage{amssymb}
\usepackage{color}

\newtheorem{theorem}{Theorem}[section]
\newtheorem{lemma}[theorem]{Lemma}

\newtheorem{proposition}[theorem]{Proposition}

\newenvironment{proof}{\paragraph{Proof.}}{\hfill $\square$\\}
\newenvironment{proof*}{\paragraph{Proof.}}{}

\newcommand{\hk}{\hslash}

\newcommand{\alg}{\mathfrak{g}}

\newcommand{\Z}{\mathbb{Z}}
\newcommand{\N}{\mathbb{N}}
\newcommand{\Rm}{\mathbb{R}}

\newcommand{\Km}{\mathbb{K}}
\newcommand{\C}{\mathcal{C}}
\newcommand{\F}{\mathcal{F}}
\newcommand{\pr}{\partial}
\newcommand{\me}{\geqslant}

\newcommand{\bra}[1]{\left (#1\right )}
\newcommand{\brac}[1]{\left [#1\right ]}
\newcommand{\pobr}[1]{\left \{#1\right \}}

\newcommand{\pd}[2]{\frac{\partial #1}{\partial #2}}
\newcommand{\var}[2]{\frac{\delta #1}{\delta #2}}

\newcommand{\res}{{\rm res}}
\newcommand{\tr}{{\rm Tr}}
\newcommand{\e}{\mathcal{E}}
\newcommand{\pmatrx}[1]{\begin{pmatrix} #1 \end{pmatrix}}
\newcommand{\Time}{\mathbb{T}}
\newcommand{\T}{{\rm T}}
\newcommand{\ad}{{\rm ad}}
\newcommand{\arrow}{\rightarrow}

\begin{document}

\title{\bf Bi-Hamiltonian structures for integrable systems on regular time scales}

\author{B\l a\.zej M. Szablikowski$^{1,2,}$\footnote{E-mail: {\tt bszablik@maths.gla.ac.uk}}, 
Maciej B\l aszak$^{2,}$\footnote{E-mail: {\tt blaszakm@amu.edu.pl}} 
and Burcu Silindir$^{3,}$\footnote{E-mail: {\tt {\tt silindir@fen.bilkent.edu.tr}}}\\[3mm]
\small $^1$ Department of Mathematics, University of Glasgow\\
\small Glasgow G12 8QW, U.K.\\[2mm]
\small $^2$ Department of Physics, Adam Mickiewicz University\\
\small Umultowska 85, 61-614 Pozna\'n, Poland\\[2mm]
\small $^3$ Department of Mathematics, Faculty of Sciences\\
\small Bilkent University, 06800 Ankara, Turkey}

\maketitle

\begin{abstract}
A construction of the bi-Hamiltonian structures for integrable
systems on regular time scales is presented. The trace functional
on an algebra of $\delta$-pseudo-differential operators, valid on
an arbitrary regular time scale, is introduced. The linear Poisson
tensors and the related Hamiltonians are derived. The quadratic
Poisson tensors is given by the use of the recursion operators of
the Lax hierarchies. The theory is illustrated by
$\Delta$-differential counterparts of Ablowitz-Kaup-Newell-Segur
and Kaup-Broer hierarchies.
\end{abstract}

\section{Introduction}

The concept of integrable systems on regular time scales can build
bridges between field systems and lattice systems. This concept
provides us not only a unified approach to study on discrete
intervals with uniform step size (i.e. lattice $\hk\Z$) and
continuous intervals but also an extended approach to study on
discrete intervals with non-uniform step size (for instance
$q$-discrete numbers $\Km_q$) or combination of continuous and discrete
intervals.

The approach of time scales allows the unification of such classes
of nonlinear evolution equations like field soliton systems
\cite{Gelfand,Adler,Konopelchenko,OS,Blaszak}, lattice soliton
systems \cite{Kuper,Suris,bm,Oevel4}, $q$-discrete soliton systems
\cite{k,f,klr,ahm} and others. The above approach was initiated in
\cite{G-G-S} where the Gelfand-Dickey construction was extended.
The theory was further developed in \cite{BSS}, where systematic
construction of $(1+1)$-dimensional integrable systems on regular
time scales was presented, and a very effective tool, classical
$R$-matrix formalism, was utilized. The $R$-matrix formalism
provides a construction of infinite hierarchies of mutually
commuting vector fields. In \cite{BSS}, we examined the general
classes of admissible Lax operators, presented examples of
integrable systems on time scales that can be written in an
explicit form. We also explained the source of constraints first
observed in \cite{G-G-S}.

The greatest advantage of the classical $R$-matrix formalism is
that it allows the construction of the bi-Hamiltonian structures
and conserved quantities. The goal of this work is to present
bi-Hamiltonian structures for $\Delta$-differential integrable
systems on regular time scales. Thus the main result of this
article is the formulation of an appropriate trace form on the
algebra of $\delta$-pseudo-differential operators, that is valid
on an arbitrary regular time scale and in particular the
real case recovers the trace form of pseudo-differential operators
\cite{Adler}. If the appropriate constraints are taken into consideration then
the trace form recovers also the one of shift operators \cite{Suris}.

In Section \ref{a}, we give a brief review of the concept of time
scales, including $\Delta$-derivative and $\Delta$-integrals. In
Section \ref{b}, we fix the class of the $\Delta$-differential
evolution equations under consideration. Besides we define
appropriate functionals and their variational derivatives. In
Sections \ref{c} and \ref{d}, we describe the algebra of
$\delta$-pseudo-differential operators, the Lax hierarchies and
the constraints that appear naturally between the dynamical fields
of admissible finite-field Lax operators. In Section \ref{e}, in
order to find the bi-Hamiltonian structures, we introduce a trace
functional on the algebra of $\delta$-pseudo-differential
operators in terms of which we construct the linear Poisson
tensors and the related Hamiltonians. The quadratic Poisson
tensors are reconstructed in the frame of the recursion operators
\cite{BSS} of the Lax hierarchies. Finally, in Section \ref{f},
the theory is illustrated by bi-Hamiltonian formulation of
finite-field integrable hierarchies on regular time scales which
are $\Delta$-differential counterparts of
Ablowitz-Kaup-Newell-Segur (AKNS) and Kaup-Broer hierarchies.

\section{Calculus on time scales}\label{a}

A time scale $\Time$ is an arbitrary nonempty closed subset of
real numbers $\Rm$ \cite{ah,hil,boh1,boh2}. For the definition of
the derivative on time scales, we use \textit{forward} and
\textit{backward jump operators} $\sigma,\rho:\Time\to\Time$
defined by
\begin{equation*}
\sigma(x)=\inf\, \{ y \in {\mathbb T}: y> x\}\qquad \rho(x)=\sup\, \{ y \in {\mathbb T}: y < x\}.
\end{equation*}
We set in addition $\sigma(\max\Time) = \max \Time$ if there
exists a finite $\max\Time$, and  $\rho(\min\Time) = \min \Time$
if there exists a finite $\min\Time$. The jump operators $\sigma$
and $\rho$ allow the classification of points on a time scale in
the following way:  $x$ is called right dense, right scattered,
left dense, left scattered, dense and isolated if $\sigma(x)=x$,
$\sigma(x)>x$, $\rho(x)=x$, $\rho(x)<x$, $\sigma(x)=\rho(x)=x$ and
$\rho(x)<x<\sigma(x)$, respectively. Moreover, we define  the
graininess function $\mu:\Time\to\Time$ as follows
\begin{equation*}
\mu(x)=\sigma (x)-x.
\end{equation*}
Besides, $\Time^\kappa$ denotes a set consisting of $\Time$ except for a possible left-scattered
maximal point. Set $x_*=\min\Time$ if there exists a finite $\min\Time$, and set
$x_* = -\infty$ otherwise. Also set $x^*=\max\Time$ if there exists a finite $\max\Time$, and set $x^* = \infty$ otherwise.

Let $f:\Time\to\Rm$ be a function on a time scale $\Time$. Delta derivative of $f$ at $x\in\Time^\kappa$, denoted by $\Delta f(x)$, is defined as
\begin{equation*}
\Delta f(x) = \lim_{s\to x} \frac{f(\sigma (x))-f(s)}{\sigma
(x)-s},\qquad s\in\Time,
\end{equation*}
provided that the limit exists. A function on a time scale is said
to be $\Delta$-smooth if it is infinitely $\Delta$-differentiable
at all points from $\Time^\kappa$.

If functions $f,g:\Time\to\Rm$ are $\Delta$-differentiable, then their product is also $\Delta$-differentiable and the following Lebniz-like rule holds
\begin{equation}\label{leib}
   \begin{split}
   \Delta(f g)(x) &= g(x)\Delta f(x) + f(\sigma(x))\Delta g(x)\\
                  &= f(x)\Delta g(x) + g(\sigma(x))\Delta f(x)
   \end{split}
   \qquad x\in \Time^\kappa.
\end{equation}
Besides, if $f$ is $\Delta$-differentiable function, then
\begin{equation}\label{rel}
    f(\sigma(x)) = f(x) + \mu(x)\Delta f(x).
\end{equation}
If $x\in\Time$ is right-dense, then $\mu(x)=0$ and the relation
\eqref{rel} is trivial.

The shift operator $E$ is defined by the formula
\begin{equation*}
Ef(x) = f(\sigma(x))\qquad x\in\Time .
\end{equation*}
Moreover the relation \eqref{rel} implies that
\begin{equation}\label{rel2}
    E = 1 +\mu\Delta.
\end{equation}

In particular, if point $x\in\Time$ lies within some continuous
interval, being part of a time scale, or if the time scale
$\Time=\mathbb{R}$, then $\Delta$-derivative is ordinary
derivative with respect to $x$, i.e. $\Delta = \pr_x$. If
$x\in\Time$ is such that $\mu(x)\neq 0$, then $\Delta =
\frac{1}{\mu}(E-1)$. This is  the case when $x$ is an isolated
point, for instance  $\Time = \Z$ or $\Km_q$.

Every continuous function $f:\Time\to\Rm$ possesses
$\Delta$-antiderivative $F:\Time\to\Rm$ such that $\Delta F(x) =
f(x)$ holds for all $x\in\Time^\kappa$. Thus we  define
$\Delta$-integral from $a$ to $b$ of $f$ by
\begin{equation}\label{int}
    \int_a^b f(x)\ \Delta x = F(b) - F(a)\qquad a,b\in\Time.
\end{equation}
Notice that, for every continuous function $f$ we have
\begin{equation}\label{sig}
    \int_x^{\sigma(x)} f(x)\ \Delta x = \mu(x)f(x).
\end{equation}
Hence, it is clear that the $\Delta$-integral is determined by local properties
of a time scale.

In particular, when the points $a$ and $b$ lie within continuous
interval, being part of a time scale, then \eqref{int} is an
ordinary Riemann integral. If all the points between $a$ and $b$
are isolated, then $b=\sigma^n(a)$ for some $n\in\Z_+$ and
$\Delta$-integral is a sum (this follows immediately from
\eqref{sig}), i.e.
\begin{equation*}
    \int_a^b f(x)\ \Delta x = \sum_{i=1}^{n-1}\mu(\sigma^i(a))f(\sigma^i(a)).
\end{equation*}
For more complicated time scales which are combinations of
continuous intervals, isolated points, etc., the integrals can be
constructed by appropriate gluing of Riemann integrals and sums.

The integration by parts formula follows from the Leibniz-like
rule \eqref{leib} as
\begin{equation}\label{part}
    \int_a^b f(x)\Delta g(x)\ \Delta x = f(x)g(x)|_a^b -
    \int_a^b g(\sigma(x))\Delta f(x)\ \Delta x,
\end{equation}
where $f$ and $g$ are continuous functions. The generalization of
\eqref{int} to the improper integral is clear. Thus, we define
$\Delta$-integral over an whole time scale $\Time$ by
\begin{equation*}
    \int_\Time f(x)\ \Delta x := \int_{x_*}^{x^*} f(x)\ \Delta x
    = \lim_{x\to x^*}F(x) - \lim_{x\to x_*}F(x)
\end{equation*}
provided that this integral converges, i.e. the limits exist.

For our purposes, we demand  the time scales where the forward
jump operator $\sigma:\Time\to\Time$ is invertible. A time scale
$\Time$ is called regular if $\sigma(\rho(x))=x$ and
$\rho(\sigma(x))=x$ for all $x\in\Time$. The first condition
implies that $\sigma$ is 'onto' and the second condition implies
that $\sigma$ is 'one-to-one'. Thus on a regular time scale
$\sigma^{-1}(x) = \rho(x)$. Actually, a time scale is regular if
and only if each point of $\Time\setminus \{x_*,x^*\}$ is either
two-sided dense or two-sided scattered and the point
$x_*=\min\Time$ is right dense and the point $x^*=\max\Time$ is
left-dense \cite{G-G-S}. For instance $\Time = [-1,0] \cup
\{1/k:k\in\mathbb{N}\} \cup \{k/(k+1):k\in\mathbb{N}\}\cup[1,2]$
is a regular time scale.

Let us consider some particular examples of regular time scales:
\paragraph{The real case, $\Time=\Rm$.} We have  $\sigma(x) = x$ and $\mu(x) =0$
for all $x\in\Rm$. In this case $\Delta$-derivative and $\Delta$-integral
are such that
\begin{equation*}
    \Delta f(x) = \pr_xf(x)\qquad\text{and}\qquad\int_{\Rm} f(x)\ \Delta x
= \int_{-\infty}^{+\infty} f(x)\ dx.
\end{equation*}
\paragraph{The lattice case, $\Time = \hk\Z$.} Let
 $\hk$ be a positive parameter. In this case $\sigma(x) = x +
\hk$ and $\mu(x) = \hk$, where $x\in\hk\Z$. $\Delta$-derivative
and $\Delta$-integral have the form
 \begin{equation*}
    \Delta f(x) = \frac{1}{\hk}\bra{f(x+\hk)-f(x)}\qquad\text{and}\qquad\int_{\hk\Z} f(x)\ \Delta x
= \hk\sum_{n\in\Z} f(n\hk).
\end{equation*}
\paragraph{The $q$-discrete numbers, $\Time = \Km_q := q^{\Z}\cup \{0\}$ ($q>1$).}
For $x\in\Km_q$, one finds that $\sigma(x) = qx$ and $\mu(x) =
(q-1)x$. Then
\begin{equation*}
    \Delta f(x) = \frac{f(qx)-f(x)}{(q-1)x}
\end{equation*}
where $x\neq 0$, and
\begin{equation*}
    \int_{\Km_q} f(x)\ \Delta x = \sum_{n\in\Z}q^{n}(q-1) f\bra{q^n}.
\end{equation*}

\section{$\Delta$-differential systems}\label{b}

Consider $N$-tuple $u:=(u_1,\ldots,u_N)^\T$ of dynamical fields $u_k:\Time\to\Rm$ being $\Delta$-smooth functions on a regular time scale $\Time$. Let
\begin{equation*}
    \C = \pobr{\Lambda u_k : k=1,\ldots,N; \Lambda\in S},
\end{equation*}
where
\begin{equation*}
    S = \pobr{\Delta^{i_1}{\Delta^\dag}^{j_1}\cdot\ldots\cdot\Delta^{i_n}{\Delta^\dag}^{j_n}: n\in\N_0, i_1, j_1\ldots,i_n, j_n\in\N},
\end{equation*}
and $\Delta^\dag$ is defined by \eqref{dag}. Therefore $S$ is the
set of all possible strings of $\Delta$ and $\Delta^\dag$
operators. Note that $\Delta$ and $\Delta^\dag$ do not commute.

What we mean by a $\Delta$-differential system, is  a system of
evolution equations
\begin{equation}\label{vec}
    u_t = K [u],
\end{equation}
where $t\in \mathbb{R}$ is an evolution parameter (time),
$u_t:=\pd{u}{t}$ and $K:= (K_1, K_2, \ldots)^\T$ with $K_i$ being
finite order polynomials of elements from $\C$, with coefficients that
might be time independent ($\Delta$-smooth) functions.

Additionally, we assume that all fields $u$ with their
$\Delta$-derivatives are rapidly decaying functions as $x$ goes to
$x_*$ or $x^*$. Then, the functionals have the following form
\begin{equation}\label{fun}
    F(u) = \int_\Time f[u]\ \Delta x,
\end{equation}
where $f[u]$ are polynomial functions of $\C$. Clearly two
densities give equivalent functionals \eqref{fun} if they differ
modulo exact $\Delta$-derivatives. Having defined the class of
evolution systems \eqref{vec} and functionals \eqref{fun}, we
further proceed in a standard way, that is  we define the duality
map, Poisson tensors, etc.

The integration by parts formula \eqref{part} leads us to the
relation
\begin{equation}\label{ap}
    \int_\Time \Delta (f) g\ \Delta x  = -\int_\Time f \Delta E^{-1}(g)\ \Delta x
    =: \int_\Time f \Delta^\dag (g)\ \Delta x.
\end{equation}
Thus the adjoint of $\Delta$-derivative is given by
\begin{equation}\label{dag}
    \Delta^\dag = -\Delta E^{-1}.
\end{equation}
Note that
\begin{equation*}
    E^{-1} = 1 + \mu\Delta^\dag.
\end{equation*}
Besides, by the use of  \eqref{rel2}, one finds that
\begin{equation}\label{adj2}
    (E\mu)^\dag = \mu(1+\mu\Delta)^\dag = \mu - \mu\Delta E^{-1}\mu
    = \mu - (E-1)E^{-1}\mu = E^{-1}\mu.
\end{equation}

Consequently, the variational derivative of a functional in the form \eqref{fun} is defined by
\begin{equation}\label{var}
    \var{F}{u_k} = \sum_{\Lambda\in S} \Lambda^\dag \pd{f[u]}{(\Lambda u_k)}\qquad k=1,\ldots,N.
\end{equation}
Notice that $\var{}{u}\Delta =0$, therefore the definition of
variational derivative \eqref{var} is consistent with the
definition of functionals \eqref{fun}.

\section{$\delta$-pseudo-differential operators}\label{c}

We introduce $\delta$ operator acting on $\Delta$-smooth functions $u:\Time\to\Rm$
by
\begin{equation}\label{delta}
\delta u:=\Delta u + Eu\delta,
\end{equation}
which is consistent with the Leibniz-like rule \eqref{leib}. By
the use of \eqref{delta}, we have
\begin{align*}
    \delta^{-1}u &= E^{-1}u\delta^{-1} + \delta^{-1}\Delta^\dag u\delta^{-1}\\
                 &= E^{-1}u\delta^{-1} + E^{-1}\Delta^\dag u\delta^{-2}
                 + E^{-1}{\Delta^\dag}^2 u\delta^{-3} + \ldots\ .
\end{align*}
The generalized Leibniz rule for the $\delta$-pseudo-differential operators takes the form
\begin{equation}\label{gl}
    \delta^n f = \sum_{k=0}^\infty S_k^{n}f\delta^{n-k}\qquad n\in\Z,
\end{equation}
where
\begin{equation*}
    S_k^n = \Delta^k E^{n-k} + \ldots + E^{n-k}\Delta^k\qquad\text{for}\qquad n\me k\me 0,
\end{equation*}
is a sum of all possible strings of length $n$, containing exactly $k$ times $\Delta$ and $n-k$ times $E$;
\begin{equation*}
    S_k^n = E^{-1}\bra{{\Delta^\dag}^k E^{n+1} + \ldots + E^{n+1}{\Delta^\dag}^k}
    \qquad\text{for}\qquad n<0\quad\text{and}\quad k\me 0
\end{equation*}
consists of the factor $E^{-1}$ times the sum of all possible strings of length $k-n-1$, containing exactly $k$ times $\Delta^\dag$ and $-n-1$ times $E^{-1}$; in all remaining cases $S_k^n = 0$. Besides, we have
the recurrence relations
\begin{equation}\label{rec1}
    S_k^{n+1} = S_k^n E + S_{k-1}^n \Delta\qquad\text{for}\qquad n\me 0
\end{equation}
and
\begin{equation}\label{rec2}
    S_k^{n-1} = \sum_{i=0}^k S_{k-i}^n E^{-1}{\Delta^\dag}^i\qquad\text{for}\qquad n<0 .
\end{equation}

\begin{lemma}\label{lemma}
For all $n\in\Z$, the relation
\begin{equation}\label{rel1}
    \sum_{k\me 0} (-\mu)^k S_k^n = (E-\mu\Delta)^n = 1
\end{equation}
holds.
\end{lemma}
The proof is postponed to Appendix.

When $x\in\Time$ is a dense point, i.e. $\mu(x)=0$, then the rule
\eqref{gl} is of the form
\begin{equation}\label{gl1}
    \delta^n f = \sum_{k=0}^\infty \binom{n}{k}\Delta^kf\delta^{n-k}\qquad n\in\Z,
\end{equation}
where $\binom{n}{k}$ is a binomial coefficient such that
$\binom{n}{k} = \frac{n(n-1)\cdot\ldots\cdot(n-k+1)}{k!}$, and
particularly when $x$ is inside of some interval then $\Delta =
\pr_x$. Thus, in this case we recover the generalized Leibniz
formula for pseudo-differential operators.

For $x\in\Time$ such that $\mu(x)\neq0$ it is more convenient to
deal with the operator $\xi := \mu\delta$ instead of $\delta$. By
the use of \eqref{delta}, the generating rule yields
\begin{equation*}
    \xi u = (E-1)u + Eu\xi,
\end{equation*}
and hence it follows that
\begin{equation}\label{gl2}
    \xi^n f = \sum_{k=0}^\infty \binom{n}{k}(E-1)^kE^{n-k}f\xi^{n-k}\qquad n\in\Z.
\end{equation}
The important fact is that the operator
$A=\sum_ia_i\delta^i$ has a unique $\xi$-representation
$A=\sum_ia'_i\xi^i$, and nonnegative (negative) order terms with
respect to $\delta$ transform into nonnegative (negative) order terms
with respect to $\xi$.

\section{Lax hierarchies}\label{d}

We define the associative algebra of $\delta$-pseudo-differential
operators
\begin{equation*}
 \alg= \alg_{\geqslant k}\oplus \alg_{< k}= \pobr{\sum_{i\geqslant k}u_{i}(x)\delta^{i}}\oplus
 \pobr{\sum_{i<k}u_{i}(x)\delta^{i}},
\end{equation*}
equipped with a Lie bracket given by the commutator $[A,B] = AB-BA$,
where $A,B\in\alg$. The classical $R$-matrices following from the
decomposition of $\alg$ into Lie subalgebras are
\begin{equation}\label{rmat}
    R = \frac{1}{2}(P_{\me k} - P_{< k}) = P_{\me k} - \frac{1}{2} = \frac{1}{2} - P_{< k},
\end{equation}
where $k=0$ or $1$ and projections are such that
\begin{equation}\label{pro}
    A_{\me k} = \sum_{i\me k} a_i\delta^i\qquad\text{for}\qquad A = \sum_i a_i\delta^i.
\end{equation}

According to the classical $R$-matrices \eqref{rmat} we have two Lax hierarchies of commuting
evolution equations \cite{BSS}
\begin{equation}\label{laxh}
       L_{t_n} = \brac{\bra{L^\frac{n}{N}}_{\me k},L}\qquad k=0,1\quad n\in\N,
\end{equation}
generated, in general, by fractional powers of some Lax operator
$L$. The general admissible finite-field
 Lax operators have the form \cite{BSS}
\begin{equation}\label{lax}
    L = u_N\delta^N + u_{N-1}\delta^{N-1}+ \ldots +u_1\delta + u_0 + \delta^{-1} u_{-1} + \sum_s\psi_s\delta^{-1}\varphi_s,
\end{equation}
where for $k=0$ the field $u_N$ is a nonzero time-independent
field and $u_{-1}=0$. Analysing \eqref{laxh} one finds that
\begin{equation*}
   \left . (-\mu)^{N+k-1}L_{t_n}\right |_{\delta=-\frac{1}{\mu}} =0,
\end{equation*}
for details see \cite{BSS}. Hence, there arises natural constraint
between dynamical fields from  \eqref{lax} given by
\begin{equation}\label{const}
    \sum_{i=-k}^{N+k-1}(-\mu)^{N+k-1-i} u_i + (-\mu)^{N+k}\sum_s \psi_s\varphi_s = a,
\end{equation}
where $a$ is time-independent function (for $k=1$ nonzero when $\mu=0$)\footnote{Notice
that in the formulae of Theorem 3.5 from \cite{BSS} there are misprints
in degrees of powers of $\mu$.}. Notice
that, the constraint \eqref{const} is compatible with the dynamics of
Lax hierarchies \eqref{laxh}. Using \eqref{const} one can
eliminate one dynamical-field. The convenient choice is to
eliminate the field $u_{N-1}$ for $k=0$ and the field $u_N$ for
$k=1$.

It is clear that in the case of $\Time = \Rm$ the Lax hierarchies
\eqref{laxh} yield field soliton systems and in particular for Lax
operators in the form \eqref{lax} one recovers the results of
\cite{Konopelchenko,OS}. In the case of $\Time=\Z$ one obtains
lattice soliton systems that are equivalent to the ones considered
in \cite{Oevel4}. Notice that in \cite{Oevel4} the $R$-matrix for
$k=1$ has a slight different form, this follows from the fact that
the construction in \cite{Oevel4} is by means of shift operators.
Similarly for $\Time =\Km_q$ the Lax hierarchies leads to
$q$-discrete soliton systems, etc.

\section{Hamiltonian structures}\label{e}

Let $A=\sum_ia_i\delta^i$ be a $\delta$-pseudo-differential
operator. We define the trace form by
\begin{equation}\label{tr}
    \tr A := -\int_\Time \frac{1}{\mu}\left. A_{<0}\right |_{\delta=-\frac{1}{\mu}}\ \Delta x\equiv
    \int_\Time \sum_{i<0}(-\mu)^{-i-1}a_i\ \Delta x,
\end{equation}
where $A_{<0}$ is the projection onto negative terms. To show that
the substitution $\delta=-\frac{1}{\mu}$ in \eqref{tr} is
well-posed, we state  the following proposition.

\begin{proposition}
Let the $\delta$-differential operators $A$ and $B$ be such that $(AB)_{<0} = AB$, then
\begin{equation}\label{sub}
    \int_\Time \frac{1}{\mu}\left. AB\right |_{\delta=-\frac{1}{\mu}}\ \Delta x
    = \int_\Time \frac{1}{\mu}\left. A\right |_{\delta=-\frac{1}{\mu}}\left.B\right |_{\delta=-\frac{1}{\mu}}\ \Delta x .
\end{equation}
\end{proposition}
\begin{proof}
It is enough to consider $A=a\delta^m$ and $B=b\delta^n$ such that $m+n<0$, thus
\begin{align*}
   \tr (AB) &= -\int_\Time \frac{1}{\mu}\left. a\delta^m b\delta^n\right |_{\delta=-\frac{1}{\mu}}\ \Delta x
   = -\int_\Time \frac{1}{\mu}\left. a \sum_{k\me 0} S_k^m b\delta^{m+n-k}\right |_{\delta=-\frac{1}{\mu}}\ \Delta x\\
    &=  \int_\Time a \sum_{k\me 0}(-\mu)^{k-m-n-1}S_k^m b\ \Delta x =
    \int_\Time ab(-\mu)^{-m-n-1}\ \Delta x,
\end{align*}
where the equality \eqref{rel1} is used. Consequently \eqref{sub} follows.
\end{proof}

Roughly speaking, the above proposition implies that the
multiplication operation in the algebra $\alg$ of
$\delta$-pseudo-differential operators commutes with the
substitution $\delta = -\frac{1}{\mu}$ given in the  trace form
\eqref{tr}.

Note that, in particular for the points $x\in\Time$ such that
$\mu(x)= 0$, the trace form \eqref{tr} turns out to be
\begin{equation*}
    \tr A = \int_\Time a_{-1}\ \Delta x,
\end{equation*}
where $A=\sum_ia_i\delta^i$. Thus when $\Time=\Rm$, we recover the
trace formula for the algebra of pseudo-differential operators
\cite{Adler}. For the case $\mu(x)\neq 0$, the trace form
\eqref{tr} within the algebra of $\xi$-operators is given by
\begin{equation*}
    \tr A := -\int_\Time \frac{1}{\mu}\left. A_{<0}\right |_{\xi=-1}\ \Delta x\equiv
    -\int_\Time \frac{1}{\mu}\sum_{i<0}(-1)^ia'_i\ \Delta x,
\end{equation*}
where $A=\sum_ia'_i\xi^i$.

\begin{theorem}
The inner product on $\alg$ defined by the bilinear map
\begin{equation}\label{inner}
    (\cdot,\cdot)_\alg:\alg\times\alg\to\Km\qquad (A,B)_\alg:= \tr (AB),
\end{equation}
in terms of the trace \eqref{tr}, is nondegenerate, symmetric and
$\ad$-invariant, i.e.
\begin{equation*}
    \bra{A,[B,C]}_\alg + \bra{[B,A],C}_\alg = 0.
\end{equation*}
\end{theorem}
\begin{proof}
The nondegeneracy of \eqref{inner} follows immediately from the definition of the trace.

In order to show that \eqref{tr} is symmetric, it is enough to
consider the monomials $A=a\delta^m$ and $B=b\delta^n$. Then, if
$m,n\me 0$, we have $\tr (AB) = \tr (BA) = 0$. If $m, n<0$ the
symmetricity immediately follows  from \eqref{sub}. Thus, it
remains to prove the case when one of the operators
 $A$ and $B$ is of positive order and the other one is of negative order.

Without loss of generality, let $m>0$ and $n<0$. We consider the
cases $\mu(x)=0$ and $\mu(x)\neq 0$, separately.

For $\mu(x) = 0$, we have
\begin{align*}
    \tr (AB) = \tr \bra{a\delta^m b\delta^n} = \tr \bra{\sum_{k=0}^m\tbinom{m}{k}a\Delta^k b\delta^{m+n-k}} = \int_\Time \tbinom{m}{m+n+1}a\Delta^{m+n+1}b\ \Delta x.
\end{align*}
The converse formula for \eqref{gl1} has the form
\begin{equation}\label{conv}
 u\delta^n = \sum_{k=0}^\infty \delta^{n-k}\binom{n}{k}{\Delta^\dag}^k u.
\end{equation}
Hence
\begin{align*}
    \tr (BA) &= \tr \bra{b\delta^n a\delta^m} = \tr \bra{\sum_{k=0}^m\tbinom{m}{k}b\delta^{m+n-k}{\Delta^\dag}^k a} = \int_\Time \tbinom{m}{m+n+1}b{\Delta^\dag}^{m+n+1}a\ \Delta x\\
    &= \int_\Time \tbinom{m}{m+n+1}a\Delta^{m+n+1}b\ \Delta x = \tr (AB),
\end{align*}
where we make use of \eqref{ap}.

For $\mu(x)\neq 0$, we pass to the calculations in terms of
$\xi$-pseudo-differential operators. Let $A=a\xi^m$ and $B=b\xi^n$
with $m>0$ and $n<0$. We have
\begin{align*}
    \tr (AB) &= \tr \bra{a\xi^m b\xi^n} = \tr \bra{\sum_{k=0}^m\tbinom{m}{k}a(E-1)^kE^{m-k}b\xi^{m+n-k}}\\
    &= -\int_\Time \frac{1}{\mu}\sum_{k=m+n+1}^m\tbinom{m}{k}(-1)^{m+n-k}a(E-1)^kE^{m-k}b\ \Delta x.
\end{align*}
The converse formula for \eqref{gl2} is given by
\begin{equation*}
  f\xi^n = \sum_{k=0}^\infty \xi^{n-k}\binom{n}{k}\bra{E^{-1}-1}^kE^{k-n}f
\end{equation*}
Let $f(E)$ be a polynomial function of $E$. Then by \eqref{adj2},
it follows that
\begin{equation*}
    \bra{\frac{1}{\mu}f\bra{E}}^\dag = \frac{1}{\mu}f\bra{E^{-1}}.
\end{equation*}
Therefore
\begin{align*}
    \tr (BA) &= \tr \bra{b\xi^n a\xi^m} = \tr \bra{\sum_{k=0}^m\tbinom{m}{k}b\xi^{m+n-k}\bra{E^{-1}-1}^kE^{k-m}a}\\
    &= -\int_\Time \frac{1}{\mu}\sum_{k=0}^m\tbinom{m}{k}(-1)^{m+n-k}b\bra{E^{-1}-1}^kE^{k-m}a\ \Delta x\\
    &= -\int_\Time \frac{1}{\mu}\sum_{k=m+n+1}^m\tbinom{m}{k}(-1)^{m+n-k}a(E-1)^kE^{m-k}b\ \Delta x
    = \tr (AB).
\end{align*}

The symmetricity of the trace functional on the algebra of
$\xi$-pseudo-differential operators implies the symmetricity of
the trace functional on the algebra of
$\delta$-pseudo-differential operators for  $\mu(x)\neq 0$.

Hence, the inner product \eqref{inner} is symmetric. Finally, the
adjoint invariance of \eqref{inner} follows from the fact that the
inner product \eqref{inner} is symmetric and the multiplication
operation defined on the algebra $\alg$ of
$\delta$-pseudo-differential operators is associative.
\end{proof}

For the next proposition, we have to consider $(1+\mu\delta)^{-1}$
and its expansion such that it is valid for all points of $\Time$
including the case $\mu=0$. Besides, we assume that the expansion
of $(1+\mu\delta)^{-1}$ is given by nonnegative order terms in
$\delta$-pseudo-differential operators. We derive the following
expansion
\begin{equation}\label{exp}
    (1+\mu\delta)^{-1} := \sum_{k=0}^\infty(-\delta)^k(\mu^k+\Delta\mu^{k+1}) \equiv
    \sum_{k=0}^\infty(-\delta)^k\frac{E\mu^{k+1}}{\mu},
\end{equation}
which can be verified by multiplying both sides of the expression
\eqref{exp} with $(1+\mu\delta)$ from right-hand side. Hence
\begin{align*}
      &(1+\mu\delta)^{-1}(1+\mu\delta) = \sum_{k=0}^\infty(-\delta)^k\frac{E\mu^{k+1}}
      {\mu} + \sum_{k=0}^\infty (-\delta)^k E\mu^{k+1}\delta\\
    &\qquad= \sum_{k=0}^\infty(-\delta)^k\frac{E\mu^{k+1}}{\mu} - \sum_{k=0}^\infty (-\delta)^{k+1}\mu^{k+1}
    - \sum_{k=0}^\infty (-\delta)^k\Delta\mu^{k+1}\\
    &\qquad= \frac{E\mu}{\mu} - \Delta\mu + \sum_{k=1}^\infty(-\delta)^k\bra{\frac{E\mu^{k+1}}{\mu}
    - \mu^k - \Delta\mu^{k+1}} = 1,
\end{align*}
where \eqref{rel2} and the converse formula \eqref{conv} are used.

\begin{proposition}
The alternative formula of the trace form \eqref{tr} is given by
\begin{equation}\label{tr2}
    \tr A = \int_\Time \frac{E^{-1}\mu}{\mu}\res\bra{A(1+\mu\delta)^{-1}} \Delta x,
\end{equation}
where
\begin{equation*}
    \res A := a_{-1}\qquad\text{for}\qquad A = \sum_i a_i\delta^i.
\end{equation*}
\end{proposition}
\begin{proof}
We  first calculate the residue. Thus
\begin{align*}
        &\res\bra{A(1+\mu\delta)^{-1}} = \res \bra{\sum_{k=0}^\infty\sum_i(-1)^ka_i\delta^{i+k}\frac{(E\mu)^{k+1}}{\mu}}\\
         &\qquad= \res \bra{\sum_{i<0}(-1)^{-i-1}a_i\delta^{-1}\frac{(E\mu)^{-i}}{\mu} + \ldots}
         = \res \bra{\sum_{i<0}(-1)^{-i-1}a_i\frac{\mu^{-i}}{E^{-1}\mu}\delta^{-1} + \ldots}\\
         &\qquad= -\sum_{i<0}\frac{(-\mu)^{-i}}{E^{-1}\mu}a_i.
\end{align*}
Substituting the residue into the alternative trace form
\eqref{tr2}, we obtain
\begin{equation*}
    \tr A = \int_\Time \frac{E^{-1}\mu}{\mu}\res\bra{A(1+\mu\delta)^{-1}} \Delta x
    = \int_\Time \sum_{i<0}(-\mu)^{-i-1}a_i\ \Delta x.
\end{equation*}
Hence the trace forms  \eqref{tr} and \eqref{tr2} are equivalent.
\end{proof}

Notice that, the definition of the trace form given by \eqref{tr2}
is very similar to the trace formula introduced in \cite{G-G-S}.
However the trace form in \cite{G-G-S} is valid only either on
$\Time=\mathbb{R}$ or on regular-discrete time scales, i.e. on
such regular time scales where all points are isolated and
$\mu\neq 0$. One of the main contribution of this article is to
generalize the trace form given in \cite{G-G-S}. The trace form
\eqref{tr} is valid for an arbitrary regular time scale. In
particular when $\Time=\mathbb{R}$, the trace form \eqref{tr}
implies the standard trace formula for the algebra of
pseudo-differential operators. Thus, for the use of the
alternative trace form \eqref{tr2}, we have to choose that $(1+\mu\delta)^{-1}$
expands into nonnegative order $\delta$-operators \eqref{exp} and
the expansion is also valid when $\mu=0$ .

Observe that, one can define alternatively the following trace
form
\begin{equation}\label{atr}
    \tr' A := \int_\Time \frac{1}{\mu}\left. A_{\me 0}\right |_{\delta=-\frac{1}{\mu}}\ \Delta x\equiv
    -\int_\Time \sum_{i\me 0}(-\mu)^{-i-1}a_i\ \Delta x,
\end{equation}
valid on regular-discrete time scales, that is for $\mu\neq 0$. Choosing
the expansion of $(1+\mu\delta)^{-1}$ into negative order terms,

\begin{equation*}
    (1+\mu\delta)^{-1} := -\sum_{k=1}^\infty(-\delta)^{-k}\frac{1}{\mu E\mu^{k-1}},
\end{equation*}
and  $\mu\neq 0$, the formula \eqref{tr2} yields \eqref{atr}, i.e.
one recovers from \eqref{tr2} the trace formula of shift
operators, as in \cite{G-G-S}.

The shift operator can be introduced by the relation: $\e = 1 +
\mu\delta$, then $\e^mu = E^mu\e^m$ for $m\in\Z$. The expansion of
the operator $A$ by means of shift operators $\e$, i.e. $A =
\sum_i a'_i\e^i$ (we assume that $\delta^{-1}$ expands into
negative order terms of shift operator $\e$) allows us to obtain
from \eqref{atr} the standard trace form of the algebra of shift
operators
\begin{equation*}
    \tr' A := \int_\Time \frac{1}{\mu}a'_0\ \Delta x.
\end{equation*}

Although the traces \eqref{tr} and \eqref{atr} are not equivalent
in general, they are closely related to each other on
regular-discrete time scales. To be more precise, on
regular-discrete time scales when applied to the constrained
operators such that $\left. A\right |_{\delta=-\frac{1}{\mu}} =
const$, the traces \eqref{tr} and \eqref{atr} are equal up to a
constant, as
\begin{equation*}
\left. A_{\me 0}\right |_{\delta = -\frac{1}{\mu}} = -\left.
A_{<0}\right |_{\delta=-\frac{1}{\mu}} + const.
\end{equation*}
Notice that, this is the case of the Lax operators \eqref{lax} for
which the constraints \eqref{const} are taken into consideration.
Thus it is clear that in the case of lattice time scale, i.e.
$\Time = \Z$, the Hamiltonians defined below, as well as the
following scheme, are equivalent to the one from \cite{Oevel4}. By
similar observations one also finds that for $\Time = \Km_q$ one
recovers from \eqref{tr2} and \eqref{atr} the trace form of
$q$-discrete numbers, see the appendix of \cite{klr}.

To sum up, the trace form \eqref{tr} is valid on arbitrary regular
time scales and in particular for $\Time=\Rm$ give the standard
form of pseudo-differential operators. Besides, if the appropriate
constraints are taken into consideration, \eqref{tr} also recovers
the trace forms for $\Time=\Z$ of 'lattice' shift operators and
for $\Time=\Km_q$ of $q$-discrete numbers.

In order to define the Hamiltonian structures for \eqref{laxh} we
first need to find adjoint of $R$-matrices \eqref{rmat}, i.e
$R^\dag $, such that $(A,R B)_\alg = (R^\dag A,B)_\alg$. Using the
alternative definition of the trace form \eqref{tr2} one finds
that
\begin{equation}\label{ar}
    R^\dag = P_{\me k}^\dag - \frac{1}{2}\qquad k=0,1\ ,
\end{equation}
where
\begin{equation}\label{ar1}
    P_{\me k}^\dag A = \bra{A(1+\mu\delta)^{-1}}_{<-k}(1+\mu\delta)
\end{equation}
and the projections are such that
\begin{equation*}
    B_{< -k} = \sum_{i<-k}\delta^ib_i\qquad\text{for}\qquad B = \sum_i \delta^i b_i.
\end{equation*}
Notice that the above projections are defined on the operators
given in a different form than in \eqref{pro}.

The existence of the well-defined inner product \eqref{inner}
allows us to identify $\alg$ with its dual~$\alg^*$. Let
$\F(\alg\cong \alg^*)$ be the space of smooth function on $\alg$
consisting of functionals \eqref{fun}. The linear Poisson tensor
has the form \cite{Oevel1,Blaszak}
\begin{equation}\label{lin}
\begin{split}
    \pi_0 dH &= \brac{R dH, L} + R^\dag\brac{dH,L}\\
             &=  \brac{L, dH_{<k}} + \bra{\brac{dH,L}(1+\mu\delta)^{-1}}_{<-k}(1+\mu\delta)\qquad k=0,1,
\end{split}
\end{equation}
where $H\in \F(\alg)$.

We do not present explicit form of the differentials $dH$ with
respect to Lax operators \eqref{lax} as it would be  cumbersome,
but we  explain how to construct them. We postulate that
\begin{equation*}
    dH = \sum_{i=1}^n \delta^{i-N-k}\gamma_i,
\end{equation*}
where $n$ is the number of independent dynamical fields in
\eqref{lax} (the number of the rest of the dynamical fields after
taking the constraint \eqref{const} into consideration). Thus, we
look forward to express $\gamma_i$'s in terms of dynamical fields
of \eqref{lax} and their variational derivatives by the use of the
assumption
\begin{equation}\label{rule}
    \bra{dH, L_t}_\alg = \int_\Time \bra{\sum_{i-k}^{N+k-2} \var{H}{u_i}(u_i)_t
    + \sum_s \bra{\var{H}{\psi_s}(\psi_s)_t + \var{H}{\phi_s}(\psi_s)_t }}\Delta x .
\end{equation}

Some useful formulae to calculate linear Poisson tensors $\pi_0$
for the examples are derived as
\begin{align*}
    P_{\me 0}^\dag(a\delta^{-1}b) &= a\delta^{-1}b + \mu ab\\
    P_{\me 1}^\dag(a\delta^{-1}b) &= a\delta^{-1}b - \delta^{-1} ab\\
    P_{\me 1}^\dag(\delta^{-1}a\delta^{-1}b) &= \delta^{-1}a\delta^{-1}b + \delta^{-1} \mu ab.
\end{align*}

The case of the quadratic Poisson tensor $\pi_1$ is more delicate
and for the construction  there appears additional conditions on
$R$ and $R^\dag$ \cite{Oevel4,Suris,Oevel1,Blaszak}. The form of
\eqref{ar} with \eqref{ar1} does not allow us to proceed in a
standard way anymore, thus we omit the construction of the
quadratic Poisson tensor from the $R$-matrix scheme. In article
\cite{BSS}, we have constructed recursion operators $\Phi$ for the
Lax hierarchies \eqref{laxh}, such that
\begin{equation*}
    \Phi L_{t_{n}} = L_{t_{n+N}}.
\end{equation*}
Thus, as we know $\pi_0$, the quadratic Poisson tensor $\pi_1$ can be reconstructed
alternatively by
\begin{equation*}
    \pi_1 = \Phi \pi_0 .
\end{equation*}
The recursion operator $\Phi$ is hereditary at least on the vector
space spanned by the symmetries from the related Lax hierarchy.
Therefore the Poisson tensors $\pi_0$ and $\pi_1$ are compatible
\cite{Olver,Blaszak}.

Hence, the Lax hierarchies \eqref{laxh} have bi-Hamiltonian structure
\begin{equation*}
    L_{t_n} = \pi_0 dH_n = \pi_1 dH_{n-N},
\end{equation*}
where the related Hamiltonians are given by
\begin{equation*}
    H_n(L) = \frac{N}{n+N}\tr \bra{L^{\frac{n}{N}+1}}.
\end{equation*}
They are such that $dH_n = L^\frac{n}{N}$.

\section{Examples}\label{f}

\paragraph{$\Delta$-differential AKNS, $k=0$.}

The Lax operator \eqref{lax} for $N=1$, with the constraint
\eqref{const} ($a=0$) is of the form
\begin{equation}\label{a0}
    L = \delta + \mu \psi\varphi + \psi\delta^{-1}\varphi.
\end{equation}

The first and the second flows from the Lax hierarchy \eqref{laxh}
are
\begin{equation}\label{a1}
    \begin{split}
    \psi_{t_1} &= \mu\psi^2\varphi+\Delta\psi,\\
    \varphi_{t_1} &= -\mu\varphi^2\psi - \Delta^\dag\varphi.
    \end{split}
\end{equation}
and
\begin{equation}\label{a2}
    \begin{split}
    \psi_{t_2} &= \mu^2\psi^3\varphi^2 + 2\psi^2\varphi + \Delta^2\psi + \Delta\bra{\mu\psi^2\varphi}
    + 2\mu\psi\varphi\Delta\psi + \mu\psi^2\Delta^\dag\varphi\\
    \varphi_{t_2} &= -\mu^2\psi^2\varphi^3 - 2\psi\varphi^2 - {\Delta^\dag}^2\varphi - \Delta^\dag\bra{\mu\psi\varphi^2} - \mu\varphi^2\Delta\psi - 2\mu\psi\varphi\Delta^\dag\varphi .
    \end{split}
\end{equation}

For Lax operator \eqref{a0} the differential of an functional $H$, such that \eqref{rule}
is valid, is given by
\begin{equation*}
    dH = \frac{1}{\varphi}\var{H}{\psi} - \frac{1}{\psi}\Delta^\dag\bra{\frac{1}{\varphi}}
    \Delta^{-1}A - \delta\frac{1}{\psi\varphi + \mu\psi\Delta^\dag\varphi}\Delta^{-1}A,
\end{equation*}
where
\begin{equation*}
    A = \psi\var{H}{\psi} - \varphi\var{H}{\varphi}
\end{equation*}
and $\Delta^{-1}$ is a formal inverse of $\Delta$.
Then, one finds the linear Poisson tensor \eqref{lin}
\begin{equation*}
    \pi_0 = \pmatrx{0 & 1\\ -1 & 0}.
\end{equation*}
The recursion operator constructed in \cite{BSS} has the form
\begin{equation*}
    \Phi =
    \pmatrx{\Delta + 2\mu\psi\varphi + 2\psi\Delta^{-1}\varphi & \mu\psi^2 + 2\psi\Delta^{-1}\psi\\
    -\mu\varphi^2 - 2\varphi\Delta^{-1}\varphi & \Delta^\dag - 2\varphi\Delta^{-1}\psi}.
\end{equation*}
Hence, the quadratic Poisson tensor is
\begin{equation}\label{quadratic}
    \pi_1 = \Phi\pi_0 =
    \pmatrx{-\mu\psi^2 - 2\psi\Delta^{-1}\psi & \Delta + 2\mu\psi\varphi + 2\psi\Delta^{-1}\varphi\\
     -\Delta^\dag + 2\varphi\Delta^{-1}\psi &  -\mu\varphi^2 - 2\varphi\Delta^{-1}\varphi }.
\end{equation}
The skew-symmetricity of  \eqref{quadratic} follows from
\eqref{rel2} and \eqref{dag}. The  first three Hamiltonians are
\begin{align*}
    H_0 &= \int_\Time \psi\varphi\ \Delta x\\
    H_1 &= \int_\Time \bra{\frac{1}{2}\mu\psi^2\varphi^2 + \varphi\Delta\psi} \Delta x\\
    H_2 &= \int_\Time \bra{\frac{1}{3}\mu^2\psi^3\varphi^3 + \psi^2\varphi^2 + \varphi\Delta^2\psi
       + \mu\psi\varphi^2\Delta\psi + \mu\psi^2\varphi\Delta^\dag\varphi} \Delta x\\
        &\ \ \vdots\ .
\end{align*}

Particularly, when $\Time=\Rm$ the above bi-Hamiltonian hierarchy
is exactly the bi-Hamiltonian field soliton AKNS hierarchy
\cite{OS}. In this case the first nontrivial flow is the second
one \eqref{a2}, i.e. the AKNS system. When $\Time=\Z$ and
$\Time=\Km_q$ we get the lattice and $q$-discrete counterparts of
the AKNS hierarchy where the first nontrivial flow is \eqref{a1}.
Besides, for $\Time=\Z$ the system \eqref{a1}, together with its
bi-Hamiltonian structure, is equivalent to the system considered
in \cite{Oevel4}.

\paragraph{$\Delta$-differential Kaup-Broer, $k=1$.}

Consider the following Lax operator with the constraint
\eqref{const} ($a=1$)
\begin{equation*}
L= (1+\mu v-\mu^2w)\delta + v + \delta^{-1}w.
\end{equation*}

The first and the second flows are
\begin{equation}\label{k1}
    \begin{split}
     v_{t_1}&=(1+\mu v-\mu^2w)\Delta v - \mu \Delta^\dag (w+\mu vw-\mu^2w^2),\\
     w_{t_1}&= -\Delta^\dag\bra{w+\mu vw -\mu^2w^2}
    \end{split}
\end{equation}
and
\begin{equation}\label{k2}
    \begin{split}
    v_{t_2}&= u \Delta\bra{v^2 + 2\tilde{u}w + \tilde{u}\Delta v + \mu \Delta^\dag(\tilde{u}w)}
    - \mu\Delta^\dag\bra{2\tilde{u}vw + \mu\tilde{u}w\Delta v + u\Delta^\dag(\tilde{u}w)}\\
    w_{t_2}&= -\Delta^\dag\bra{2\tilde{u}vw + \mu\tilde{u}w\Delta v + \tilde{u}\Delta^\dag(\tilde{u}w)},
    \end{split}
\end{equation}
where
\begin{equation*}
    \tilde{u} := 1+\mu v-\mu^2w.
\end{equation*}

The differentials are given in the form
\begin{equation*}
    dH = \delta^{-1}\var{H}{v} + \var{H}{w} + \mu \var{H}{v}.
\end{equation*}
Thus,  the linear Poisson tensor \eqref{lin} is
\begin{equation*}
    \pi_0 = \pmatrx{\tilde{u}\Delta\mu - \mu\Delta^\dag\tilde{u} & \tilde{u}\Delta\\ -\Delta^\dag \tilde{u} & 0}.
\end{equation*}
The recursion operator has the form \cite{BSS}
\begin{equation*}
 \Phi =   \pmatrx{w + \tilde{u}\Delta + R & \mu v - \mu w + (2+\mu\Delta^\dag)\tilde{u}- R\mu\\
    w -\Delta^\dag\tilde{u}w\Delta^{-1}\tilde{u}^{-1} & \Delta^\dag\tilde{u} + v - \mu w \Delta^\dag\tilde{u}w\Delta^{-1}\mu\tilde{u}^{-1}},
\end{equation*}
where
\begin{equation*}
        R = \tilde{u} \Delta v\Delta^{-1}\tilde{u}^{-1} -\mu\Delta^\dag\tilde{u} w\Delta^{-1}\tilde{u}^{-1}.
\end{equation*}
Hence
\begin{equation*}
    \pi_1 = \Phi\pi_0 =
    \pmatrx{ \pi_{vv} & \tilde{u}\Delta v + \tilde{u}\Delta \tilde{u}\Delta
    + \mu \tilde{u}w\Delta - \mu\Delta^\dag\tilde{u}w\\
           -v\Delta^\dag\tilde{u} -\Delta^\dag\tilde{u}\Delta^\dag\tilde{u}
           + \tilde{u}w\Delta\mu - \Delta^\dag\mu\tilde{u}w  & \tilde{u}w\Delta - \Delta^\dag\tilde{u}w},
\end{equation*}
where
\begin{equation*}
    \pi_{vv} = \tilde{u}\Delta\mu v - \mu v\Delta^\dag \tilde{u} + \tilde{u}\Delta\tilde{u} - \tilde{u}\Delta^\dag\tilde{u} + \tilde{u}\Delta\tilde{u}\Delta\mu - \mu\Delta^\dag\tilde{u}\Delta^\dag\tilde{u} + \mu\tilde{u}w\Delta\mu - \mu\Delta^\dag\mu\tilde{u}w .
\end{equation*}

The Hamiltonians are
\begin{align*}
    H_0 &= \int_\Time w\ \Delta x\\
    H_1 &= \int_\Time \bra{vw - \frac{1}{2}\mu w^2} \Delta x\\
    H_2 &= \int_\Time \bra{w^2 + v^2w + (w+\mu vw-\mu^2w^2)\Delta v - \frac{2}{3}\mu^2 w^3} \Delta x\\
        &\ \ \vdots\ .
\end{align*}

When $\Time=\Rm$ the above construction recovers the field
Kaup-Broer hierarchy with its bi-Hamiltonian structure
\cite{Konopelchenko}. As previously, the Kaup-Broer system is
given only by the second flow \eqref{k2}. In the lattice case,
i.e. of $\Time=\Z$, the above bi-Hamiltonian hierarchy is
equivalent to the relativistic Toda hierarchy considered in
\cite{Oevel4} and, \eqref{k1} is equivalent to the relativistic
Toda system.

\section{Conclusions}

We have presented a unified theory of the construction of the
bi-Hamiltonian nonlinear evolution hierarchies such as field,
lattice and $q$-discrete soliton hierarchies. Actually, we took
advantage from the theory of time scales. Therefore, one can also
consider the construction of soliton systems with spatial variable
belonging to the spaces being partially continuous and discrete.
This might be interesting from the point of view of applications.
There are also other approaches generalizing and unifying theory
of soliton systems, presented in \cite{DM-H} and \cite{BGSSz}.

On the other hand, making use of the regular-discrete time scales
only, the theory from the article can be considered as a
discretization scheme of field soliton systems. In some special
cases, introducing appropriately deformation parameter to some
regular-discrete time scales, one can consider the quasi-classical
limit of discrete soliton systems yielding dispersive field
soliton equations. In particular for $\Time=\hk\Z$ the
quasi-classical limit is given by $\hk\arrow 0$ and for
$\Time=\Km_q$ by $q\arrow 1$, see \cite{BSS}.

\section*{Acknowledgement}

This work was partially  supported by the Scientific and Technical
Research Council of Turkey and MNiSW research grant no. N N202 404933.
B.Sz. was supported by the European Community under a Marie Curie
Intra-European Fellowship, contract no. PIEF-GA-2008-221624.

\section*{Appendix}

We verify the Lemma \ref{lemma} by considering the positive and
negative cases of $n$ separately by the use of induction. Note
that $E-\mu\Delta = E^{-1}-\mu\Delta^\dag = 1$.

Let $n\me 0$. Assume that \eqref{rel1} holds for positive $n$. Then
\begin{align*}
    (E-\mu\Delta)^{n+1} &= (E-\mu\Delta)^n(E-\mu\Delta) = (E-\mu\Delta)^nE - \mu(E-\mu\Delta)^n\Delta\\
    &= \sum_{k=0}^n(-\mu)^kS_k^nE + \sum_{k=0}^n(-\mu)^{k+1}S_k^n\Delta
    = \sum_{k=0}^{n+1}(-\mu)^kS_k^nE + \sum_{k=0}^{n+1}(-\mu)^kS_{k-1}^n\Delta\\
    &= \sum_{k=0}^{n+1}(-\mu)^k\bra{S_k^nE + S_{k-1}^n\Delta} = \sum_{k=0}^{n+1}(-\mu)^kS_k^{n+1},
\end{align*}
where we used the fact that $S_{n+1}^n = S_{-1}^n = 0$ and the recurrence relation \eqref{rec1}.

Let $n<0$. First we show \eqref{rec1} for $n=-1$. Thus, using the
recursive substitution we deduce
\begin{align*}
    (E-\mu\Delta)^{-1} &= \bra{E^{-1}-\mu\Delta^\dag}(E-\mu\Delta)^{-1}
           = E^{-1} - \mu(E-\mu\Delta)^{-1}\Delta^\dag\\
           &= E^{-1} - \mu \bra{E^{-1}-\mu(E-\mu\Delta)^{-1}\Delta^\dag}\Delta^\dag
           = E^{-1} - \mu E^{-1}\Delta^\dag + \mu^2(E-\mu\Delta)^{-1}{\Delta^\dag}^2\\
           &= E^{-1} - \mu E^{-1}\Delta^\dag + \mu^2E^{-1}{\Delta^\dag}^2 - \mu^3E^{-1}{\Delta^\dag}^3 + \ldots\\
           &= \sum_{k=0}^\infty(-\mu)^kE^{-1}{\Delta^\dag}^k = \sum_{k=0}^\infty(-\mu)^kS_k^{-1}.
\end{align*}
Assume  that \eqref{rel1} holds for negative $n$. Then
\begin{align*}
    (E-\mu\Delta)^{n-1} &= (E-\mu\Delta)^n(E-\mu\Delta)^{-1}
    = \sum_{k=0}^\infty(-\mu)^kS_k^n \sum_{i=0}^\infty(-\mu)^iE^{-1}{\Delta^\dag}^i\\
    &= \sum_{k=0}^\infty\sum_{i=0}^\infty(-\mu)^{k+i}S_k^nE^{-1}{\Delta^\dag}^i = \sum_{k=0}^\infty\sum_{i=0}^k(-\mu)^kS_{k-i}^nE^{-1}{\Delta^\dag}^i\\
    &= \sum_{k=0}^\infty(-\mu)^k\sum_{i=0}^kS_{k-i}^nE^{-1}{\Delta^\dag}^i = \sum_{k=0}^\infty(-\mu)^kS_k^{n-1},
\end{align*}
where we used \eqref{rel1} for $n=-1$ and the recurrence relation
\eqref{rec2}. Hence \eqref{rel1} holds for $n-1$, which finishes
the proof.

\end{document}